\documentclass[letterpaper]{article}
\usepackage{times}
\usepackage{helvet}
\usepackage{courier}
\usepackage{url}
\usepackage{graphicx}
\usepackage{hyperref}
\usepackage{amsmath}
\usepackage{amssymb}
\usepackage{amsthm}
\usepackage{gensymb}
\usepackage{amsfonts}
\usepackage{bm}
\usepackage{multirow}
\usepackage{algorithm,algorithmicx}
\usepackage{tcolorbox}
\usepackage{authblk}
\def\argmin{\mathop{\hbox{argmin}}\limits}
\def\argmax{\mathop{\hbox{argmax}}\limits}

\def\x{{\mathbf x}}
\def\y{{\mathbf y}}
\def\z{{\mathbf z}}

\def\a{{\mathbf a}}
\def\b{{\mathbf b}}

\def\floor#1{{\left\lfloor\,#1\,\right\rfloor}}
\def\ceil#1{{\left\lceil\,#1\,\right\rceil}}

\newcommand{\sNorm }[1]{\mbox{}\|#1\|  }

\newcommand{\OsNorm }[1]{\mbox{}\|#1\|_{\ell_1}  }

\newcommand{\FsNorm }[1]{\mbox{}\|#1\|_F  }

\newcommand{\TsNorm}[1]{\mbox{}\|#1\|_2}

\newcommand{\MaxNorm }[1]{\mbox{}\|#1\|_{\max}  }

\newcommand{\setlinespacing}[1]%
           {\setlength{\baselineskip}{#1 \defbaselineskip}}

\newcommand{\abs }[1]{\left|#1\right|}
\newcommand{\sabs }[1]{|#1|}

\newtheorem{definition}{Definition}

\newtheorem{lemma}{Lemma}
\newtheorem{theorem}{Theorem}
\newtheorem{corollary}{Corollary}

\newcommand{\mat}[1]{{\ensuremath{\bm{\mathrm{#1}}}}}

\def\a{{\bm \alpha}}
\def\w{{\mathbf w}}

\def\b{{\mathbf b}}

\def\x{{\mathbf x}}
\def\y{{\mathbf y}}
\def\z{{\mathbf z}}

\def\matA{\mat{A}}
\def\matB{\mat{B}}

\def\matD{\mat{D}}

\def\matM{\mat{M}}

\def\matR{\mat{R}}

\def\matU{\mat{U}}
\def\matV{\mat{V}}
\def\matW{\mat{W}}
\def\matX{\mat{X}}
\def\matY{\mat{Y}}

\begin{document}
\title{Ternary Residual Networks}
\author[1]{Abhisek Kundu
}
\author[1]{Kunal Banerjee}
\author[1]{Naveen Mellempudi}
\author[1]{Dheevatsa Mudigere}
\author[1]{Dipankar Das}
\author[1]{Bharat Kaul}
\author[2]{Pradeep Dubey}
\affil[1]{Parallel Computing Lab, Intel Labs, Bangalore, India}
\affil[2]{Parallel Computing Lab, Intel Labs, Santa Clara, CA, USA}
\affil[ ]
{\textit{\{abhisek.kundu, kunal.banerjee, naveen.k.mellempudi, dheevatsa.mudigere, dipankar.das, bharat.kaul, pradeep.dubey\}@intel.com}}
\date{}
\maketitle
\begin{abstract}
Sub-8-bit representation of DNNs incur some discernible  loss of accuracy despite rigorous (re)training at low-precision. Such loss of accuracy essentially makes them equivalent to a much shallower counterpart, diminishing the power of being deep networks.  
To address this problem of accuracy drop we introduce the notion of \textit{residual networks} where we add more low-precision 
edges to sensitive branches of the sub-8-bit network to compensate for the lost accuracy. Further, we present a perturbation theory to identify such sensitive edges. 
Aided by such an elegant trade-off between accuracy and compute, the 8-2 model (8-bit activations, ternary weights), enhanced by ternary residual edges, turns out to be sophisticated enough to achieve very high accuracy ($\sim 1\%$ drop from our FP-32 baseline), despite $\sim 1.6\times$ reduction in model size, $\sim 26\times$ reduction in number of multiplications, and potentially $\sim 2\times$ power-performance gain comparing to 8-8 representation, on the state-of-the-art deep network ResNet-101 pre-trained on ImageNet dataset.  Moreover, depending on the varying accuracy requirements in a dynamic environment, the deployed low-precision model can be upgraded/downgraded on-the-fly by partially enabling/disabling residual connections. For example, disabling the least important residual connections in the above enhanced network, 
the accuracy drop is $\sim 2\%$ (from FP32), despite $\sim 1.9\times$ reduction in model size, $\sim 32\times$ reduction in number of multiplications, and potentially $\sim 2.3\times$ power-performance gain comparing to 8-8 representation.
Finally, all the ternary connections are sparse in nature, and the ternary residual conversion can be done in a resource-constraint setting with no low-precision (re)training. 
\end{abstract}

\section{Introduction}

Deep Neural Networks (DNNs) (AlexNet \cite{alexnet}, VGG \cite{vgg}, ResNet \cite{resnet}) achieved remarkable accuracy in many application domains, such as, image classification, object detection, semantic segmentation, speech recognition  (\cite{dnn_nature}).  
However, DNNs are 
notoriously resource intensive models in terms of amount of compute, memory bandwidth requirements, and consumption of power. 
Deploying trained DNNs to resource-constraint devices (mobile, cars, robots) to make billions of predictions every day, efficiently and accurately, with limited power budget is a considerably challenging problem.    
This motivates a compact and/or reduced-precision model of DNNs for both mobile devices and data  servers. 

There are two major approaches to such problems. One is to reduce the number of parameters in the network (e.g. finding a shallower/compact  representation) yet achieve similar accuracy as the deep network. Examples of such kind are SqueezeNet \cite{squeezenet}, MobileNet \cite{mobilenet}, and SEP-Nets \cite{sepnet}. These models are very small in size ($\sim$ 5 MB) and are typically targeted for mobile devices.
However, it is not surprising that their overall accuracy is very limited on complex dataset ImageNet \cite{imagenet_data}, e.g.,  SEP-Net Top-1 accuracy is 65.8\% which is $\sim$ 10\% off to that of ResNet-50. Deploying them on sensitive applications, such as autonomous cars and robots, might be impractical because these models might make too many mistakes (hence might be fatal as well). 
The other approach is concerned about the reduction in size of parameter representation (via compression or low-precision). Well-known methods of this kind are pruning  \cite{dnn_surgery,structured_sparsity,energy_pruning}, quantization \cite{logqunat,finn2016,gupta2015lp,inq,hubara2016qnn,NNFM2015}, binarization \cite{courbariaux2016bnn}, ternarization \cite{2016twn,zhu2016ttq,han2015learning,TWD2017,FGQ17}, hashing \cite{quant_hashing}, Huffman coding \cite{deep_compression} and others \cite{zhou2016dorefa,rastegariECCV16}.
However, despite smaller size of network representation, not all of these techniques may be friendly to efficient implementation on general purpose hardware (CPUs, GPUs) (e.g., \cite{deep_compression}). 
%
Additionally, the power consumption of DNNs 
depends mostly on the data movement of the feature maps and the number of multiply-and-accumulate (MAC), rather than model size. 
For example, convolution layers, despite having much smaller number of parameters comparing to FC layers, consume more power 
because of repeated use of convolution weights. 
Similarly, thinner and deeper networks might consume more power than shallower networks. For example, SqueezeNet \cite{squeezenet}, despite being $50\times$ smaller in size than AlexNet with comparable accuracy, consumes 33\% more power \cite{energy_pruning}. 

Here we are mainly focused on the trade-off between low-precision representation and accuracy of deeper networks, keeping an eye on the power-performance factors.  
There is a clear need for reduction in precision for both weights and activations (that are fetched from external memory or I/O devices)
for more efficient implementation of deeper networks.
Such low-precision representation for activation demands for specialized low-precision arithmetic \cite{williamson1991dynamically,williamson1985new,FGQ17} and hardware design.   
For example, Google's TPU \cite{tpuGblog} sets a trend for a low-precision inference pipeline (8-bit activations, 8-bit weights).
%
%
%
%
Moreover, significant research energy is being expended to explore sub-8-bit domain of DNN representation \cite{rastegariECCV16,zhou2016dorefa,hubara2016qnn,FGQ17}, where the interplay among model size, compute, power, and accuracy becomes more tricky. For 8-8 representation, despite reducing the model size by $4\times$, only minimal drop in accuracy has been observed for deeper networks. However, the accuracy degrades significantly in sub-8-bit domain where we reduce precisions for the weights and/or the activations. Low-precision (re)training is a popular approach to recover a part of the lost accuracy.

(Re)training at low-precision essentially re-parametrizes the DNNs to find another local optima in high-dimensional, non-convex search space of parameters. However, it is not clear if such low-precision solutions with similar generalization ability as FP-32 solutions exist, and also how to find them efficiently. In reality, sub-8-bit models for deep networks incur some noticeable drop in accuracy. This loss severely undermines the purpose of deploying a deep (sub-8-bit) network. For example, an 8-4 (4 bit weights) model, if produces $\sim 2\%$ drop on ResNet-101, might be equivalent to 8-8 model on much shallower ResNet-50 in terms of model size and accuracy, but might be worse in power-performance (for deep networks even a small gain in accuracy costs significant compute (Table \ref{table:resnet_summary})). 
\begin{table}[!t]
\centering
    \begin{tabular}{| c || c | c | c | c | c |}
   \hline
ResNets & 18 & 34 & 50 & 101 & 152 
\\ \hline \hline
Top-1 Err (\%) & 27.88 & 25.03 & 24.7 & 23.6 & 23
\\ \hline
FLOPs ($10^9\times$) & 1.8 & 3.6 & 3.8 & 7.6 & 11.3
\\ \hline
  \end{tabular}
\caption
{
Accuracy vs compute for FP-32 ResNet models on ImageNet 1K classification. For deeper models,  even a small gain in accuracy costs significant compute. 
}
\label{table:resnet_summary}
\end{table}
This weakens the motivation for sub-8-bit models of deep networks. 
%
%
We seek to answer: 
\begin{tcolorbox}
\textit{Can a sub-8-bit model achieve similar accuracy as FP-32 model with comparable or better model size and power-performance numbers comparing to 8-8}?
\end{tcolorbox}
Considering computational benefits of ternary 8-2 models, \cite{FGQ17} introduced a fine-grained quantization (FGQ) method that first partitions the weights into disjoint blocks, and then ternarizes them. The block size ($N$) controls the accuracy vs number of multiplications (and model size). Larger $N$ eliminates more multiplications, however, with a notable drop in accuracy (although they reported the best accuracy for sub-8-bit models on ImageNet).
%
Another limitation of existing models is that they cannot be set on a `power-saving' mode (say, via less MAC) when some further loss in accuracy is tolerable in a dynamic environment. Once deployed, existing models essentially operate in a `fixed-accuracy-fixed-power' mode.

To deal with the problems discussed above, we introduce the notion of \textit{residual edges} (especially) for sub-8-bit DNNs, where we add more sub-8-bit parameters to sensitive branches of the network to compensate for the lost accuracy. we propose a perturbation theory on the pre-trained DNNs to estimate the sensitivity of branches and the number of residual edges we need in order to achieve a given accuracy.
We apply this method on ternary 8-2 representation for ResNet-101 and AlexNet pre-trained on ImageNet. Guided by the theory and enhanced by the ternary residual edges, the ternary 8-2 representation turns out to be sophisticated enough to outclass the 8-8 model in terms of model size, number of multiplications, and power-performance gain, while  maintaining very high  accuracy. Moreover, such networks with residual edges can be upgraded/downgraded on-the-fly by partially enabling/disabling some of the residual edges, depending on the accuracy requirements in a dynamic environment. For example, when autonomous cars or robots are in a less eventful environment where less number of objects are involved, the classification problem becomes considerably  simpler (sufficient to distinguish among distinct objects, such as humans, dogs, vehicles, trees, etc. rather than discriminating among multiple breeds of dogs), and by disabling many edges we can downgrade the model in terms of compute, power, etc., yet maintain very high accuracy for those (less number of) classes.

Drawing an analogy between human attention and precision, and also an analogy between stress due to attention and power consumption, it is natural for us to be selectively attentive to certain tasks that requires processing more information. Such upgrade/downgrade of low-precision DNNs mimics a more real-world scenario that other existing models are unable to imitate. For example, both 8-2 ternary and 8-8 are always at fixed precision fixed power mode irrespective of the dynamic nature of the circumstances.  
Using such a downgrade operation of our residual network (for ResNet-101) we observe only $2\%$ drop in classification accuracy (from our FP-32 baseline) 
keeping only $2\times$ parameters of ternary 8-2 network, despite eliminating $\sim 32\times$  multiplications and achieving $\sim 2.3\times$ power-performance gain over 8-8. 
Finally, the conversion from FP-32 weights to 8-2 ternary residual model requires no low-precision (re)training and it can be performed in a resource-constraint environment. 
%
%

We organize the rest of the paper as follows. We interpret low-precision/sparse representation as adding  noise to pre-trained DNNs. For this, 
we first provide a perturbation analysis of pre-trained DNNs to figure out sensitivity of key quantities that contributes to the final error. Then, we introduce the notion of residual parameters which, when added to the perturbed network, reduces the noise and improves the accuracy. Specifically, we focus on ternary 8-2 models, and show that ternary residual networks can outperform 8-8 representation in terms of critical factors, such as, model size, number of multiplications, and power-performance, while maintaining very high accuracy. Finally, experiments on ResNet-101 and AlexNet for ImageNet classification problem corroborate our theoretical results. We start with summarizing the frequently-used notations below.
\\

$\bullet$ \textbf{Notations and Preliminaries: }
For a matrix $\matA\in\mathbb{R}^{m\times n}$, we denote the (element-wise) Frobenius norm as $\FsNorm{\matA} = \sqrt{\sum_{i,j=1}^{m,n}\sabs{\matA_{ij}}^2}$, $\ell_1$ norm as $\OsNorm{\matA}=\sum_{i,j=1}^{m,n}\sabs{\matA_{ij}}$, and max norm as $\MaxNorm{\matA} = \max_{i,j}\sabs{\matA_{ij}}$. We can easily generalize these definitions to higher order tensors. Also, we define the spectral/operator norm as $\TsNorm{\matA} = \sup\{\FsNorm{\matA\x}/\FsNorm{\x}:\x\in\mathbb{R}^n, \x\neq \textbf{0}\}$. For matrices $\matA$ and $\matB$, $\FsNorm{\matA\matB}\leq \min\{\FsNorm{\matA}\TsNorm{\matB},\TsNorm{\matA}\FsNorm{\matB}\}$. For vectors $\x$ and $\y$ of same dimension, inner product is defined as $\left<\x, \y\right>= \sum_i\x_i\cdot \y_i$. From Cauchy-Schwarz inequality, $\sabs{\left<\x, \y\right>}\leq \FsNorm{\x}\FsNorm{\y}$. 
Also, $\OsNorm{\matA}\leq \sqrt{mn} \cdot \FsNorm{\matA}$.

\section{Perturbation in a Locally-Optimal DNN} \label{sec:perturbation_bound}
We first provide an analysis on how the output of a pre-trained (locally optimal) DNN gets distorted, layer-by-layer, in presence of additive noise to inputs and/or parameters. We want to control some of the key components of the noise to put a limit on this overall perturbation such that it has little or no impact on the accuracy. 
%
We treat a DNN as a composition of parametric transformation functions $f_i$, and we can get (locally) optimal values for the parameters $\matW^{(i)}$ via network training (optimizing some parametric function defined over input space $\mathcal{D}_0$).   
Then, we can interpret quantization and/or sparsification as adding a noise to the (locally) optimal parameters $\matW^{(i)}$ to produce sub-optimal $\tilde\matW{}^{(i)}$. We want to quantify the effect of such sub-optimal $\tilde\matW{}^{(i)}$ on the final outcome, e.g., classification scores produced by the DNN. 
%
For this, let us assume that our DNN has $\ell$ layers and let $\y \in \mathbb{R}^d$ ($d$ is the number of classes) be the output vector of layer $\ell$ such that its $i$-th component $\y_i$ contains the classification score for $i$-th class, for a given input $\x$. 
Let $\hat\y \in \mathbb{R}^d$ denote the perturbed vector $\y$ due to added noise.
%
Here we are interested in top-1 accuracy, i.e., the largest component of $\y$ should remain the largest in $\hat\y$. Mathematically, let $i^*$ be the index for the correct class label for an input, and we define $i^*$ and $j^*$ as: 
$
i^* = \argmax_i \{\y_i\}
$, and 
$
j^* = \argmax_j \{\hat\y_j\}
$.  
Then, $i^* = j^*$ implies no loss of classification accuracy despite the perturbation in $\hat \y$. 
Note that,
$
\FsNorm{\y - \hat \y} \leq \delta 
\Rightarrow \max_i\abs{\y_i - \hat \y_i} \leq \delta,
$
for $\delta>0$,
i.e., $\FsNorm{\y - \hat \y} \leq \delta$ implies that no component of the original output vector 
can be altered by more than $\pm\delta$. 
For a correctly classified input $\x$, a misclassification occurs due to perturbation when $\hat\y_j >\hat\y_{i^*}$, for $j\neq i^*$. 
Let us assume that, due to perturbation, the true classification score gets reduced and some other score gets bigger, i.e.,  
$\hat\y_{i^*} = \y_{i^*}- \gamma\cdot\delta$,  and 
$\hat\y_{j} = \y_{j}+\sqrt{1-\gamma^2}\cdot \delta$, for $i^*\neq j$ and $0\leq \gamma \leq 1$, such that $\FsNorm{\y - \hat \y} \leq \delta$.
Then, such $\delta$-perturbation does not cause misclassification if
$\hat\y_{i^*} > \hat\y_{j}$, i.e., $\y_{i^*} - \y_j > \gamma\cdot \delta + \sqrt{1-\gamma^2}\cdot \delta \ge \delta$. 
In other words, 
as long as the true classification score is at least $\delta$ higher than any other score, then a $\delta$-perturbation has no adverse effect on classification accuracy.
We want to construct a $\hat\y$ (e.g., based on sparse and/or low-precision representation of weights/activations) such that $\FsNorm{\y-\hat\y} < \delta$, for a given tolerance $\delta$ and for any input $\x$. We can choose $\delta = \y_{i^*} - \y_j$ for this example. In reality, we can choose $\delta$ more judiciously from a distribution of such differences on training data. 
For  better interpretation and simplicity of the analysis we consider the relative perturbation as follows.
\begin{eqnarray*}
\frac{\FsNorm{\y-\hat\y}}{\FsNorm{\y}} 
\leq
\varepsilon
\leq
\frac{\FsNorm{\y-\hat\y}}{\sabs{\y_{i^*}}}
<
\frac{\delta}{\sabs{\y_{i^*}}}
=
\frac{\y_{i^*} - \y_j}{\sabs{\y_{i^*}}}
\end{eqnarray*}
We want to first derive an upper bound on 
${\FsNorm{\y-\hat\y}}/{\FsNorm{\y}}$ in terms of layer-wise parametric and/or non-parametric perturbation of the pre-trained network, and then we want to control such perturbations to keep the final accumulated  perturbation to be smaller than $\varepsilon$ (which can be chosen according to a distribution of such relative difference on training data).
We now define a set of functions used in a DNN. 
\begin{eqnarray}\label{eqn:f_dnn}
f_{dnn} = \{
\text{Convolution with Batch Normalization, 
Matrix Multiplication, 
ReLU, Pooling
}
\}
\end{eqnarray}
Functions in $f_{dnn}$ can be linear or non-linear, parametric or non-parametric, convex or non-convex, smooth or non-smooth. A pre-trained DNN is a fixed composition of functions in $f_{dnn}$ with locally optimal parameter values, i.e., DNN: $f=f_\ell f_{\ell-1}...f_1 : \mathcal{D}_0\rightarrow \mathcal{D}_\ell$, where each $f_i \in f_{dnn}$. 
%
More explicitly, let $f_i: \mathcal{D}_{i-1} \rightarrow \mathcal{D}_i$ with parameters $\matW^{(i)}$, and each $\mathcal{D}_i$ is  an arbitrary metric space where $\sNorm{\cdot}$ denotes a distance metric on set $\mathcal{D}_i$ (for simplicity, we focus on normed space only). 
For  all $\matX^{(0)} \in \mathcal{D}_0$, we define $\matX^{(i)} \in \mathcal{D}_i$ and $\tilde\matX{}^{(i)} \in \mathcal{D}_i$ as follows. For $i=1,...,\ell$,
$$
\matX^{(i)} = f_i(\matX^{(i-1)};\matW^{(i)}), 
\quad
\tilde\matX{}^{(i)} = f_i(\tilde\matX{}^{(i-1)};\tilde\matW{}^{(i)}),
$$
where $\tilde\matX{}^{(i)}$ and $\tilde\matW{}^{(i)}$ are perturbed versions of $\matX^{(i)}$ and $\matW^{(i)}$, respectively.
We want to measure how the outcome of $f$ gets perturbed in presence of perturbed inputs $\tilde\matX{}^{(i)}$ and perturbed parameters $\tilde\matW{}^{(i)}$. More specifically, we want to quantify the relative change in outcome of a given layer:
$
{\sNorm{\matX{}^{(i)}-\tilde\matX{}^{(i)}}}/{\sNorm{\matX{}^{(i)}}}
$. 
We note that input to a layer might be perturbed due to perturbation in earlier layers and/or perturbation in the present layer (e.g., activation quantization) before it is applied to the layer function. For this we use separate notations as follows. For $i$-th layer, let $\tilde\matX{}^{(i-1)}$ denote the perturbed input, $\hat\matX{}^{(i-1)}$ denote the perturbed activation, and $\tilde\matW{}^{(i)}$ denote the perturbed weights. 
Let us first define the following relative perturbations.
\begin{eqnarray}\label{eqn:epsilon}
\Delta_i =  
\frac{\FsNorm{\matX^{(i)}- \tilde\matX{}^{(i)}}}{\FsNorm{\matX^{(i)}}}
, 
\quad
\gamma_i = \frac{\FsNorm{\tilde\matX{}^{(i)}- \hat\matX{}^{(i)}}}{\FsNorm{\matX^{(i)}}}
, 
\varepsilon_i = 
\frac{\FsNorm{\matW^{(i)} - \tilde\matW{}^{(i)}}}{\FsNorm{\matW^{(i)}}}
\end{eqnarray}
We derive the following result to bound the relative change in the output of a layer using the definitions in (\ref{eqn:epsilon}).
\begin{theorem}\label{thm:layer_bound}
Using the above notations, the relative change in output of $i$-th layer of DNN, can be bounded as
\begin{eqnarray}\label{eqn:layer_bound}
\nonumber
\Delta_i \quad 
\leq
\quad
(\prod_{k=1}^iO(1+\varepsilon_k))\Delta_0
&+&
\sum_{k=1}^i(\prod_{j=k+1\leq i}^iO(1+\varepsilon_j))\text{ }O(\gamma_{k-1})
\\
&+&
\sum_{k=1}^i(\prod_{j=k+1\leq i}^i\text{  }O(1+\varepsilon_j))O(1+\gamma_{k-1})O(\varepsilon_k)
\end{eqnarray}
\end{theorem}
The theorem can be proved using triangle inequality, recursion in $\Delta_i$ in (\ref{eqn:epsilon}), and with the assumption that locally optimal parameters are not in a neighborhood of zero.
Theorem \ref{thm:layer_bound} gives us an upper bound on how much the output of $f_i$ changes by the perturbation of parameters, perturbation of activations, and perturbation of the domain of the composition ($\Delta_0$).
The result suggests that at $i$-th layer of DNN, perturbations of parameters and activations of all the previous stages accumulate nonlinearly in a weighted manner.
Moreover, we can see that perturbations in early layers accumulate more to create larger final perturbation, suggesting higher  sensitivity of earlier layers.
We want this perturbation to be small such that the perturbed solution remains in a neighborhood of the local optima (and generalize in a similar manner). 
We simplify the above bound for small noise.
\begin{theorem}\label{cor:small_perturbation}
Using the above notations, assuming $\Delta_0=0$ and $\FsNorm{\hat\matX{}^{(i-1)}}\leq \tau_{i-1}\FsNorm{\matX^{(i-1)}}$, where $\tau_{i-1}>0$ are constants, we derive the following for constants 
$c_{j}>0$. 
\begin{eqnarray}\label{eqn:final}
\Delta_i
&\leq&
\sum_{k=1}^i(\prod_{j=k+1\leq i}^ic_j)(O(\gamma_{k-1})+
O(\varepsilon_k))
\end{eqnarray}
\end{theorem}
The result suggests that layer-by-layer small amount of changes in weights and activations accumulate in a weighted manner. 
That is, keeping both $\gamma_i$ and  $\varepsilon_i$ small implies overall small perturbation, and as long as this is smaller than the relative gap between top score (for correct classification) and the next best score, there would be no loss in classification accuracy.
It is intuitive that larger the gap between the true classification score and the next best score more perturbation a DNN can tolerate.
Also, for earlier layers $\varepsilon_k$ should be kept much smaller than those in later layers to have an overall small perturbation. Our empirical evaluation on quantized DNNs (ResNet-101 and AlexNet) corroborates this theory. 
 %
\section{Low Precision DNN}
We interpret quantization, sparsification, low-precision, etc. as adding a noise to the locally optimal parameters of a DNN. Such noisy solutions, despite showing some degradation in accuracy, can be computationally beneficial. 
We want to choose a noise model carefully such that the noisy solution can be implemented efficiently on general purpose hardware.  
The focus here is to find a low-precision representation of DNNs that can benefit from low-precision operations while maintaining extremely high accuracy.

$\bullet$ \textbf{8-bit Activations and 8-bit Weights:} 
Constraining activations and weights to 8 bits 
appears to induce only small perturbation, resulting in typically $<1\%$ loss in accuracy. This can be explained using (\ref{eqn:final}) where we observe small $\varepsilon_k$ and $\gamma_k$ for each layer, such that, the final perturbation affects the relative difference between true classification scores and the next best scores minimally for the entire test set.

$\bullet$ \textbf{Sub-8-bit Representation:}
More challenging cases are sub-8-bit representation of DNNs, e.g., 8-bit activations and ternary weights. 
Note that the bound in (\ref{eqn:layer_bound}) suggests a non-linear degradation of accuracy in terms of two errors: error in activations and error in weights. That is, keeping both of them at very low precision implies a large amount of degradation of classification scores. This is likely  because the perturbed solution stays away from a neighborhood of the local optima. In such cases, we typically need to find another local optima via (re)training at low-precision \cite{FGQ17}. 
This new local optima need not be in a neighborhood of the old optima. In fact, there could be multiple optima mapping to similar accuracy  via re-parametrization \cite{sharp_minima}. However, it is not clear if low-precision solutions exist which can show very similar accuracy as the full-precision ones. Moreover, finding such solutions in a very-high dimensional, non-convex search space might be a non-trivial task. In reality, we often observe a noticeable drop in accuracy for such sub-8-bit solutions (despite rigorous (re)training at low precision). One possible explanation could be that these sub-8-bit models have limited model capacity (e.g., the number of unique values they can represent). 
We can interpret the earlier layers of a DNN as features generation stage, followed by feature synthesis, and finally the classification/discrimination phase. Lack of bits in early layers might severely constrain the quality of features generated, and consequently, more sophisticated features at later stages become coarse, degrading the quality of the network. This intuitive explanation is also consistent with the theoretical bound in (\ref{eqn:layer_bound}), where perturbation in earlier layers gets magnified more. It is natural to demand for more accuracy in low-precision representation (in robotics, autonomous cars, etc.), and the existing  methods may be insufficient to deal with this problem. 
There is a need to understand the trade-off between accuracy and precision in a theoretically consistent way. This paper attempts to  address such a case.


\subsection{Ternary Conversion of Pre-Trained Parameters}{\label{sec_ternaryconversion}}
Here we consider one specific sub-8 bit representation of DNN: 8-bit activations and ternary weights (8-2), where we want to decompose the full-precision trained weights $\matW$ to ternary values $\{-\alpha, 0, +\alpha\}$, $\alpha\geq 0$,  without re-training. We consider a simple threshold ($T>0$) based approach similar to \cite{2016twn,FGQ17}. Let $\hat\matW$ denote a ternary weight, such that, $i$-th element $\hat\matW_i = sign(\matW_i)$, if $\sabs{\matW_i}> T$, and $0$ otherwise.
Then, the element-wise error is  $E(\alpha, T)=\FsNorm{\matW-\alpha\hat\matW}^2$ and an optimal ternary representation $\alpha^*\hat\matW{}^*\approx \matW$ is as follows:
\begin{eqnarray}\label{eqn:twn}
\alpha^*, T^* = \argmin_{\alpha\geq 0, T>0}E(\alpha, T),
\text{ s.t. }
\alpha \ge 0, \hat\matW_i\in \{-1,0,+1\}
\end{eqnarray} 
for $i=1,2,...,n$, where $\matW \in \mathbb{R}^n$.
%
However, 
the weights may learn different types of features and 
may follow different distributions. Combining all the weights together might represent a mixture of various distributions, and a ternary representation for them,  
using a single threshold ($T$) and magnitude ($\alpha$), may not preserve the distributions of the weights.
%
To deal with this problem \cite{FGQ17} introduced FGQ by first dividing these weight tensors into disjoint blocks of sub-tensors of size $N$, and then ternarizing such blocks independently, i.e., decomposing $\matW$ into a disjoint group of $k$ filters $\{\matW^{(i)}\}, i=1,...,k$, and corresponding ternary weights $\alpha_i\hat\matW{}^{(i)}$, where  $\hat\matW{}^{(i)}_j\in \{-1,0,+1\}, \forall j$, solve $k$ independent sub-problems.
\begin{eqnarray}\label{eqn:main}
\alpha_1^*,..,\alpha_k^*,\hat\matW{}^{(1)}{}^*,..,\hat\matW{}^{(k)}{}^*
=
\sum_i
{\argmin_{\alpha_i,\hat\matW{}^{(i)}}}
\FsNorm{\matW^{(i)}-\alpha_i\hat\matW{}^{(i)}}^2
\end{eqnarray}
%
Denoting $I_T = \{i:\sabs{\matW_i}>T\}$, 
optimal solutions to individual sub-problems can be derived as
\begin{eqnarray}\label{eqn:symmetric}
\alpha^* 
= (\sum_{i\in I_T}\sabs{\matW_i})/\sabs{I_T}
,
T^* = \argmax_{T>0}(\sum_{i\in I_T}\sabs{\matW_i})^2/\sabs{I_T}
\end{eqnarray}
This leads to overall smaller layer-wise $\ell_2$ error; consequently, FGQ shows better accuracy using a smaller $N$. This improvement in accuracy is consistent with the theory in (\ref{eqn:final}). From model capacity point of view, with $k$ disjoint ternary FGQ blocks we can represent up to $2k+1$ distinct values, i.e., model capacity increases linearly with number of such blocks. However, smaller $N$, despite showing lower $\ell_2$ error, leads to larger number of (high-precision) multiplications (larger number of $\alpha$'s), and this might lead to less efficient implementation on general purpose hardware. 
%

For very deep networks, e.g., ResNet-101, we need significantly larger number of fine-grained parameters (synthesized from early layers) to improve the accuracy even by a small margin from its shallower counterparts (Table \ref{table:resnet_summary}). 
%
Sub-8-bit representation of sensitive parameters may have a `blurring' effect on later activations; consequently, extremely high accuracy results might be elusive in 8-2 model.

\subsection{Ternary Residual Edges} 
\label{sec: residual}

Motivated by achieving extremely high accuracy using sub-8-bit representation/operations, we introduce the notion of Residual Edges: when $\ell_2$ error between original weights and low-precision weights is high we need additional (sub-8) bits to compensate for the lost accuracy. That is, for sensitive branches of network we add more sub-8-bit edges to maintain the desired model capacity. This takes the final solution to a neighborhood of original solution.

Inferencing with parametric functions in $f_{dnn}$, such as convolution and matrix multiplication, can be expressed as a linear operation. For a given input $\x$, (partial) output can be expressed as $\y = \matW\x$, where $\matW$ are learned weights. Clearly, $\y = \tilde\matW\x+(\matW-\tilde\matW)\x$, where $\tilde\matW$ is some perturbed version of $\matW$. In our ternary setting, let $\tilde\matW = \alpha\hat\matW$, where $\alpha\hat\matW$ is a ternary representation of $\matW$, via Algorithm \ref{alg:ternary}. Let, the residual be $\Delta = \matW-\alpha\hat\matW$. For any given input if we accumulate the (partial) outputs of both the ternary weights and the residual weights, then we recover the original output. 
%
However, the residual $\Delta$ may not be low-precision. In order to have a uniform low-precision operation, such as 8-2, we need to approximate $\matW$ as a sequence of ternary residuals, such that, accumulating the output of all these intermediate steps gets us closer to the original output. Let, $\tilde\matW_0$, $\tilde\matW_1$, ..., $\tilde\matW_r$, be a sequence of ternary weights, where $\tilde\matW_0=\text{Ternary}(\matW)$, first step residual is $\Delta_1=\matW-\tilde\matW_0$,  $\tilde\matW_1=\text{Ternary}(\Delta_1)$, $\Delta_2=\matW-(\tilde\matW_0+\tilde\matW_1)$, ..., $\Delta_r = \matW-\sum_{i=0}^{r-1}\tilde\matW_i$,  $\tilde\matW_r=\text{Ternary}(\Delta_r)$. The (ternary) inference on input $\x$ is $\tilde\y=\sum_i\tilde\matW_i\otimes\x$, where $\otimes$ denotes the ternary multiplication. The goal is to ensure $\tilde\y \approx \y=\matW\x$.
%
Accumulation of such (ternary) residuals is guided by the perturbation theory presented here 
where we need to preserve the output of each layer in order to maintain a small perturbation in the final outcome.
Before we specify the steps of our ternary residual algorithm more formally, we need a more in-depth comparison with FGQ approach. 
%
\subsubsection{Comparison with FGQ Ternary 
}
We can represent only three distinct values: $-\alpha, 0, +\alpha$ with ternary weights. Both FGQ and our residual method increase model capacity, i.e., the number of distinct values that can be represented using them. With $k$ FGQ blocks we can have up to $2k+1$ distinct values, i.e., model capacity increases linearly with the number of blocks. However, this produces multiple scaling factors $\alpha$ that leads to larger number of multiplications (typically inefficient). On the other hand, with $r$ step ternary residual we can represent up to $3^{r+1}$ distinct values, that is, model capacity increases exponentially. However, residual approach results in an increased  model size (linear in $r$) as we need to store $r+1$ number of ternary weights. 
We can alleviate this problem by combining FGQ with residual ternary. That is, we can apply ternary residual for each ternary FGQ block.  
Moreover, not all the blocks are equally important, and we might need different number of residuals for different blocks. Let $i$-th block requires $r_i$ number of residuals to approximate it up to some desired accuracy. Then, there will be total $\sum_{i=1}^k(r_i+1)$ scaling factors, model capacity can be expressed as $\sum_i3^{(r_i+1)}-k+1$, and model size (in bits) is $(8+\frac{2n}{k})\sum_{i=1}^k(r_i+1)$. Table \ref{table:summary_params} summarizes the comparison of various ternary methods (we assume scaling factors $\alpha$'s are 8-bit each).

\begin{table*}[!t]
\centering
    \begin{tabular}{| c || c | c | c |}
   \hline
 & Model Size & Model Capacity 
 & \# Scaling Factors
\\ \hline \hline
Ternary & $8+2n$ & $3$ & $1$ \\ \hline
FGQ Ternary ($k$ blocks) & $8k+2n$
& $2k+1$ & $k$ \\ \hline
Ternary Residual ($r$ steps) & $(r+1)(8+2n)$ & $3^{r+1}$  & $r+1$ \\ \hline
FGQ + Ternary Residual & 
$(8+\frac{2n}{k})\sum_{i=1}^k(r_i+1)$
& $\sum_{i=1}^k3^{(r_i+1)}-k+1$ 
& $\sum_{i=1}^k (r_i+1)$ \\ \hline
  \end{tabular}
\caption
{
Comparison of various ternary methods, for a vector of length $n$, in terms of number of scaling factors (typically proportionate to number of multiplications), model capacity (number of distinct values that can be represented) and model size (number of bits). We assume that each scaling factor $\alpha$ is 8-bit.  
$r_i$ denotes the number of residuals used for the $i$-th block.
}
\label{table:summary_params}
\end{table*}
%
%
In (\ref{eqn:final}), we express the final perturbation in terms of $\ell_2$ norm of layer-wise perturbation. Here we extend it to block level perturbation as follows. Let $\matW^{(j)}$ be pre-trained weight  for $j$-th layer ($\tilde\matW{}^{(j)}$ be its perturbed version). Also, let $\matW^{(j)}$ is partitioned into $k$ disjoint blocks $\matW^{(j)}_{(i)}$ ($\tilde\matW{}^{(j)}_{(i)}$ be the perturbed version), $i=1,..., k$. 
Then, sensitivity $\varepsilon^{(j)}_i$ of $i$-th block in $j$-th layer is defined as follows.
\begin{eqnarray}\label{eqn:block_sensitivity}
\text{(Block Sensitivity)} 
\quad 
\varepsilon^{(j)}_i
=
{\FsNorm{\matW_{(i)}^{(j)}-\tilde\matW{}_{(i)}^{(j)}}}/{\FsNorm{\matW^{(j)}}}
\end{eqnarray}
We relate $\varepsilon^{(j)}_i$ with $\varepsilon_j$ (defined in (\ref{eqn:epsilon})) as follows.
$$
\sum_{i=1}^k (\varepsilon^{(j)}_i)^2
=
{\sum_{i=1}^k \FsNorm{\matW^{(j)}_{(i)}-\tilde\matW{}^{(j)}_{(i)}}^2}/{\FsNorm{\matW^{(j)}}^2}
= 
\varepsilon_j^2 
$$
%
(\ref{eqn:block_sensitivity}) suggests that for a given perturbation of weights of a layer, various blocks of weights may be perturbed differently. Consequently, we might need different number of residuals for different blocks to bound the total perturbation of a given layer. We present an incremental algorithm (Algorithm \ref{alg:residual}) where we add a ternary residual to the block that creates the largest error (even after residuals have been added to it). We repeat the process until the error for the layer is below certain desired tolerance.  
Here we give a proof that adding ternary residual blocks, as in Algorithm \ref{alg:residual}, strictly reduce the $\ell_2$ error at every step. 
\begin{theorem}\label{theorem:residual}
Let $\delta^{(i)}$ denote the $\delta$ computed at the $i$-th iteration in Algorithm \ref{alg:residual}. Then,
$\delta^{(i)} < \delta^{(i-1)}$, for all $i$.
\end{theorem}
\begin{figure}[!t]
\centering
\centering
\begin{algorithm}[H]
\caption{Ternary}\label{alg:ternary}
\begin{algorithmic}[1]
\State \textbf{Input:} Weights $\matW \in \mathbb{R}^n$.
\State Find $\alpha^*$ and $\hat\matW{}^*$ using (\ref{eqn:symmetric}).
\State \textbf{Output:} Ternary vector $\{\alpha^*\hat\matW{}^*\}$.
\end{algorithmic}
\end{algorithm} 
\centering
\begin{algorithm}[H]
\caption{Ternary Residual}\label{alg:residual}
\begin{algorithmic}[1]
\State \textbf{Input:} Weights $\matW$, tolerance $\varepsilon>0$.
\State Partition $\matW$ into $K$ disjoint group of weights $\matW_{(k)}$, for $k=1,...,K$.
\State 
$\alpha^{(1)}_{(k)}\hat\matW{}^{(1)}_{(k)} \leftarrow$ Ternary($\matW_{(k)}$), for $k=1,..., K$.
\State $\Delta=\matW-\sum_k\alpha^{(1)}_{(k)}\hat\matW{}^{(1)}_{(k)}$, 
\quad $\delta = \FsNorm{\Delta}^2/\FsNorm{\matW}^2$.
\State Let $\matV$ be a list of $K$ multi-set of ternary vectors, such that, 
$\matV_{(k)} = \{\alpha^{(t)}_{(k)}\hat\matW{}^{(t)}_{(k)}\}
$, for $t=1,...,\tau_k=\sabs{\matV_{(k)}}$, $k=1,..., K$.
\State Let $E$ be a list of $K$ errors, such that, $E_k = \FsNorm{\matW_{(k)} - \sum_t\alpha^{(t)}_{(k)}\hat\matW{}^{(t)}_{(k)}}$, $t=1,...,\sabs{\matV_{(k)}}$.
\State While $\delta > \varepsilon^2$
\State \quad 
$k^* = \argmax_k \{E_k\}$. 
%
\State \quad 
$\Delta_{k^*}=\matW_{(k^*)} - \sum_{t=1}^{\sabs{\matV_{(k^*)}}}\alpha^{(t)}_{(k^*)}\hat\matW{}^{(t)}_{(k^*)}$. 
\State
\quad 
$\alpha^{(\tau_{k^*}+1)}_{(k^*)}\hat\matW{}^{(\tau_{k^*}+1)}_{(k^*)} \leftarrow$ 
Ternary($\Delta_{k^*}$).
\State \quad
$\matV_{(k^*)} \leftarrow \matV_{(k^*)} \cup \{\alpha^{(\tau_{k^*}+1)}_{(k^*)}\hat\matW{}^{(\tau_{k^*}+1)}_{(k^*)}\}$
\State \quad 
$E_{k^*} \leftarrow \FsNorm{\matW_{k^*} - \sum_{t=1}^{\sabs{\matV_{(k^*)}}}\alpha^{(t)}_{(k^*)}\hat\matW{}^{(t)}_{(k^*)}}$, 
\State \quad 
$\delta \leftarrow \FsNorm{\matW - \sum_k\sum_{t=1}^{\sabs{\matV_{(k^*)}}}\alpha^{(t)}_{(k)}\hat\matW{}^{(t)}_{(k)}}^2/\FsNorm{\matW}^2$.
\State \textbf{Output:} Ternary multi-set of vectors $\matV$.
%
\end{algorithmic}
\end{algorithm} 
\end{figure}

%
It is intuitive that when the  magnitude of a bunch of numbers follow a bi-modal distribution (with one peak is centered close to zero), a ternary representation (one zero and one non-zero magnitude) might approximate the distribution well. In this case, the  scaling factor $\alpha$ is close to the non-zero peak. However, when the magnitude of the numbers represent more than two clusters, ternary representation induces large $\ell_2$ error. We have observed that layer-wise pre-trained weights (magnitudes)  follow exponential, heavy-tailed, or half-Gaussian distributions. Therefore, ternary representation with only one $\alpha$ results in large $\ell_2$ error for weights. Consequently, the large overall  perturbation in the network leads to poor accuracy (as predicted by our theory). 
FGQ blocking is an attempt to approximate such distribution at a finer level with larger number of $\alpha$'s (at the cost of more multiplications). However, the above problem of large $\ell_2$ error resurfaces for FGQ when we ternarize larger blocks. In such case, our proposed residual approach comes to the rescue, 
where we refine the solution by adding a ternary approximation of distribution of element-wise errors. That is, poorly-approximated elements by the earlier ternary representations become more refined. As discussed earlier, FGQ increases the model capacity linearly in number of blocks, while our residual approach improves it exponentially with number of residual steps (Table \ref{table:summary_params}). We can interpret model capacity as an indicator of how many different cluster centers we can represent (thereby how well a mixture of clusters can be approximated). Then, residual approach creates exponentially more cluster centers, and it is intuitive that we can achieve a desired level of $\ell_2$ approximation with only few steps of residual. 
Using ternary residual approach, we essentially approximate an arbitrary distribution with a small sequence of bi-modal distributions. 
%
%
%

One unique property of the residual ternary representation is that we can selectively enable/disable some of the residual weights on-the-fly to upgrade/downgrade the model in response to varying accuracy requirements in a dynamic environment (e.g., autonomous cars, robots, etc.). This is unlike the existing approaches where we may not have such flexibility once we deploy the model, especially on a resource-constraint device. Disabling least important residuals can save a lot of compute while having little impact on the accuracy. 
%
We can interpret the scenario as a `battery-savings mode' on a resource-constraint device. 
Another interesting property of the ternary residual connections/blocks is that they are sparse in nature, and are highly compressible (and suitable for sparse operations). 
Finally, for ease of implementation on a general purpose hardware, the partitioning/blocking of weights are done in a memory contiguous way. That is, we can unroll the weight tensor into a vector, then pick $N$ consecutive element from the vector to form a block of weights. 
As argued by \cite{FGQ17}, $k$-means clustering of weights
might lead to better approximation, but may not be friendly to efficient implementation. 
\\
\\
\textbf{Power-Performance Estimate:} 
Let $X$ be the power-performance gain for ternary 8-2 operations over 8-8. Then, power-performance for residual method with $C\times$ compute using FGQ block size $N$ can be shown as $\frac{X}{C(X/N+1)}$.
\section{Experiments}\label{sec:experiments}
We use pre-trained FP-32 AlexNet and ResNet-101 models for ImageNet classification task.  
%
For 8-bit activation/weight quantization, we have used the low-precision dynamic fixed point technique mentioned in \cite{FGQ17}. We applied Algorithm \ref{alg:residual} to convert the FP-32 models to the ternary residual networks.
As suggested by our theory, 
earlier layers require less perturbation to control the overall error. We set gradually smaller values for layer-wise perturbation $\varepsilon_k$ for earlier layers. We note the total number of ternary blocks (FGQ+residual) required to achieve a given accuracy. Intuitively, we add more ternary compute (proportional to a factor of total blocks) in order to achieve higher accuracy.  
\\
\textbf{Compute-Aware Perturbation:} 
We can set the tolerances $\varepsilon_k$ in a compute-aware manner considering the layer-wise compute distribution to reduce overall compute for ternary residual networks (Figure \ref{fig:compute_epsilon}). 

For power-performance gain, we estimate $X\sim 5.5$ for $N=64$.
We summarize our results in Tables \ref{table:ResNet_accuracy} and \ref{table:AlexNet_accuracy}. For comparison, we mention a few results of other sub-8-bit networks on ImageNet using AlexNet. (1) Binarized weights and activations of \cite{rastegariECCV16} incurred a loss of $\sim 12\%$, (2) the loss of binary weights and 2-bit activations of \cite{zhou2016dorefa} was $\sim 6\%$ from FP-32, and (3) \cite{hubara2016qnn} with binary weight and 2-bit activations reduced the loss to $\sim 5.5\%$ from FP-32, (4) using FGQ with $N=4$ (75\% elimination of multiplications), \cite{FGQ17} achieved $\sim 7.8\%$ loss from FP-32 using ternary weights and 4-bit activations without any low-precision re-training. Note, however, that (1) and (2) used FP-32 weights and activations for the first and last layers. Finally, all these models have different power-performance profile, and a detailed analysis on this is beyond the scope of this paper. 
%
\begin{figure}[t]
\centering
\includegraphics[width=8cm,height=8cm]
{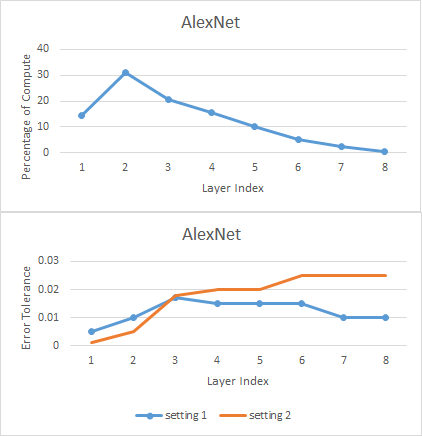}
\caption{Compute-Aware Perturbation: Incurring only $\sim1\%$ loss (from FP32 accuracy) for Ternary Residual AlexNet via compute-aware tolerance selection. Both the settings of $\varepsilon_k^2$'s lead to the same accuracy, however, with different compute profile. In setting 1, slightly larger $\varepsilon_k$'s are used for first two (sensitive) layers to gain in compute; however, to prevent further loss, we set $\varepsilon_k$'s smaller for later layers, while maintaining gain in overall compute comparing to setting 2.
}
\label{fig:compute_epsilon}
\end{figure}
\begin{table}[t]
\centering  
\begin{tabular}{c||c|c|c|c }
    \hline
    \multirow{2}{*}{} &
      \multicolumn{2}{c}{Loss $\sim 1\%$} &
      \multicolumn{2}{c}{Loss $\sim 2\%$} 
\\
    & \# Blocks & comp & \# Blocks & comp 
    \\
    \hline\hline
    $N=64$ & $2.3\times$ & $2.5\times$
    & $2\times$ & $2.2\times$
    \\
    \hline
    $N=32$ & $2.1\times$ & $2.3\times$ 
    & $1.7\times$ & $2\times$
    \\
    \hline
  \end{tabular}
\caption{Results for Ternary Residual ResNet-101 via Algorithm \ref{alg:residual} using FGQ blocks ($N$ being block size). Here we assume that the total number of blocks (and compute) for FGQ ternary (without residual) is $1\times$. Loss is w.r.t FP32.}
 \label{table:ResNet_accuracy}
\end{table}
\begin{table}[t]
\centering
  \begin{tabular}{c||c|c|c|c }
    \hline
    \multirow{2}{*}{} &
      \multicolumn{2}{c}{Loss $\sim 1\%$} &
      \multicolumn{2}{c}{Loss $\sim 2\%$} 
\\
    & \# Blocks & comp & \# Blocks & comp
    \\
    \hline\hline
    $N=64$ & $2.9\times$ & $2.9\times$
    & $2.9\times$ & $2.5\times$
    \\
    \hline
    $N=32$ & $2.6\times$ & $2.5\times$ 
    & $2.6\times$ & $2.2\times$
    \\
    \hline
  \end{tabular}
\caption{Results for Ternary Residual AlexNet via Algorithm \ref{alg:residual} using FGQ blocks ($N$ being block size). Here we assume that the total number of blocks (and compute) for FGQ ternary (without residual) is $1\times$. Loss is w.r.t FP32.}
 \label{table:AlexNet_accuracy}
\end{table}

\section{Discussion}

In order to achieve extremely high accuracy for sub-8-bit DNNs, we introduce the notion of residual inference, where we add more sub-8-bit edges to sensitive branches of the sub-8-bit network. Such addition of residual edges is guided by a perturbation theory proposed here for pre-trained DNNs. We show that ternary 8-2 models, enhanced by such ternary residual edges, can outperform 8-8 networks in terms of model size, number of multiplications, inference power-performance, while maintaining similar accuracy. A unique feature of residual enhancement is that we can upgrade/downgrade the model on the fly, depending on the varying accuracy requirements in a dynamic environment, by enabling/disabling selected branches. Moreover, 
the ternary residual network can be formed from FP-32 counterpart in a resource-constrained environment without (re)training.
%
Although we presented the residual idea only for one type of sub-8-bit representation, e.g., ternary 8-2, it is general enough to be applied for other low-precision representations as well (for both weight and activations). A future work is to combine our residual approach with low-precision (re)training in a theoretically consistent manner to improve the power-performance numbers.
\bibliographystyle{plain}
\bibliography{references_new}
\newpage
\section{Appendix}
\subsection{Proofs}
\subsubsection{Proof of Theorem \ref{thm:layer_bound}}\label{proof:thm_layer_bound}
We first bound the relative change in output in presence of perturbed input and perturbed parameters using Lemma \ref{lemma:general_bound}.
\begin{lemma}\label{lemma:general_bound}
Let $\matX^{(i-1)}$ be an input to $f_i \in f_{dnn}$ with pre-trained parameter $\matW^{(i)}$, and let $\matX^{(i)}$ be the output. Then, for perturbed input $\tilde\matX{}^{(i-1)}$, perturbed activations $\hat\matX{}^{(i-1)}$, perturbed parameter $\tilde\matW{}^{(i)}$, and perturbed output $\tilde\matX{}^{(i)}$, we derive, for constants $c_i>0$,
\begin{eqnarray}\label{eqn:general_bound}
\frac{\FsNorm{\matX^{(i)} - \tilde \matX{}^{(i)}}}{\FsNorm{\matX^{(i)}}}
\leq 
c_i\left(\frac{\FsNorm{\matX^{(i-1)}- \tilde\matX{}^{(i-1)}}}{\FsNorm{\matX^{(i)}}}
+
\frac{\FsNorm{\tilde\matX{}^{(i-1)}- \hat\matX{}^{(i-1)}}}{\FsNorm{\matX^{(i)}}}
+
\frac{\FsNorm{\hat\matX{}^{(i-1)}}}{\FsNorm{\matX^{(i)}}}\frac{\FsNorm{\matW^{(i)}- \tilde\matW{}^{(i)}}}{\FsNorm{\matW^{(i)}}}
\right)
\end{eqnarray}
For non-parametric functions the last term in (\ref{eqn:general_bound}) is zero.
\end{lemma}
Let
$$
\Delta_i =  \frac{\FsNorm{\matX^{(i)}- \tilde\matX{}^{(i)}}}{\FsNorm{\matX^{(i)}}}, 
\quad
\gamma_i =  \frac{\FsNorm{\tilde\matX{}^{(i)}- \hat\matX{}^{(i)}}}{\FsNorm{\matX^{(i)}}}, 
\quad
\varepsilon_i =  \frac{\FsNorm{\matW^{(i)}- \tilde\matW{}^{(i)}}}{\FsNorm{\matW^{(i)}}}.
$$
We can write, for some constant $0< c_1\leq 1$,
\begin{eqnarray}\label{eqn:triangle}
\FsNorm{\hat\matX{}^{(i-1)}} 
\leq
c_1\cdot\FsNorm{\matX^{(i-1)}} 
+ 
c_1\cdot\FsNorm{\matX^{(i-1)}-\tilde\matX{}^{(i-1)}}
+
c_1\cdot\FsNorm{\tilde\matX{}^{(i-1)}-\hat\matX{}^{(i-1)}}.
\end{eqnarray}
Then, combining (\ref{eqn:general_bound}) and (\ref{eqn:triangle}), we have the following recursive relation.
\begin{eqnarray*}
\Delta_i 
\leq 
c_i\cdot \Delta_{i-1} 
+
c_i\cdot \gamma_{i-1} 
+
c_i c_1(1+\gamma_{i-1}+\Delta_{i-1})\varepsilon_i
=
c_i (1+c_1\varepsilon_i)\Delta_{i-1} 
+ 
c_i\cdot \gamma_{i-1} 
+
c_i c_1(1+\gamma_{i-1})\varepsilon_i
\end{eqnarray*}
Simplifying the constants, we derive
\begin{eqnarray}\label{eqn:recursion}
\Delta_i 
\leq 
O(1+\varepsilon_i)\Delta_{i-1} 
+ 
O(\gamma_{i-1})
+
O(1+\gamma_{i-1})O(\varepsilon_i)
\end{eqnarray}
Expanding the recursion in (\ref{eqn:recursion}) we get the following. 
\begin{eqnarray*}
\Delta_i
&\leq&
\left(\prod_{k=1}^iO(1+\varepsilon_k) \right)\Delta_0
+
\sum_{k=1}^i\left(\prod_{j=k+1\leq i}^iO(1+\varepsilon_j)\right)(O(\gamma_{k-1})
+
O(1+\gamma_{k-1})O(\varepsilon_k))
\end{eqnarray*}
\subsubsection{Proof of Theorem \ref{cor:small_perturbation}}\label{sec:proof_corollary1}
In Lemma \ref{lemma:general_bound}, under small perturbation, we assume that $\FsNorm{\hat\matX{}^{(i-1)}}\leq d_{i-1}\cdot\FsNorm{\matX^{(i-1)}}$, for some constant $d_{i-1}>0$.
Then, simplifying the constants, from (\ref{eqn:general_bound}) we derive
\begin{eqnarray*}
\Delta_i
&\leq&
\left(\prod_{k=1}^ic_k\right)\Delta_0
+
\sum_{k=1}^i\left(\prod_{j=k+1\leq i}^ic_j\right)(O(\gamma_{k-1})
+
O(\varepsilon_k))
\end{eqnarray*}
For no input domain perturbation, we set $\Delta_0=0$ to derive the result.
\subsubsection{Proof of Lemma \ref{lemma:general_bound}}
We follow some notational convention. Let $\x$ denote original input, $\tilde\x$ denote the input perturbed due to perturbation of earlier layers, $\hat\x$ denote perturbed $\tilde\x$ in the current layer (say because of low precision quantization) before applying it to the layer function. Perturbed weights are denoted by $\tilde\matW$. \\
\\
$\bullet$ \textbf{Conv+BN+Scaling (Parametric)}:
\\
Convolution essentially performs an inner product between an image patch and an  weight filter. 
The $i$-th element of $k$-th outpur feature map (ofm) can be expressed as 
$$
y_i^{(k)} = \left<\matX^{(i)},\w^{(k)}\right>
$$
where $\matX^{(i)}$ is the $i$-th input patch where $k$-th ofm convolution filter bank is applied.
Also, let batch normalization (BN) and scaling parameters for $k$th ofm are $\mu_k$, $\sigma_k$, and $\alpha_k$, $\beta_k$ respectively, such that the combined output of convolution, BN, and Scaling can be expressed as
$$
z_i^{(k)} = \frac{y_i^{(k)}-\mu_k}{\sigma_k}\alpha_k + \beta_k
=\frac{\alpha_k}{\sigma_k}\cdot y_i^{(k)} + \left(\beta_k - \frac{\mu_k\cdot\alpha_k}{\sigma_k}\right)
= a_k \cdot y_i^{(k)}+b_k,
$$
where $a_k$ and $b_k$ are learned (locally optimal) parameters. 
That is, we have a linear expression
$$
z_i^{(k)} = a_k\cdot \left<\matX^{(i)},\w^{(k)}\right>+ b_k
=
\left<\matX^{(i)},a_k\cdot\w^{(k)}\right>+ b_k
= \left<\left(\matX^{(i)},1\right),\left(a_k\cdot\w^{(k)}, b_k\right)\right>
=
\left<\bar\matX{}^{(i)},\bar\w{}^{(k)} \right>
$$
Let the perturbed output be 
$$
\tilde z_i{}^{(k)} = \left<\hat \matX{}^{(i)},\tilde\w{}^{(k)} \right>
$$
where $\tilde \matX{}^{(i)}$, $\hat \matX{}^{(i)}$ and $\tilde\w{}^{(k)}$ are perturbed input patch, perturbed activation patch, and perturbed parameter.
Then, for some constant $c_1>0$, change in output can be bounded as follows.
\begin{eqnarray*}
\sabs{z_i{}^{(k)} - \tilde z_i{}^{(k)}}^2
&\leq&
c_1^2(\sabs{\left<\bar\matX{}^{(i)},\bar\w{}^{(k)} \right> - \left<\tilde\matX{}^{(i)},\bar\w{}^{(k)} \right>}^2
+
\sabs{\left<\tilde\matX{}^{(i)},\bar\w{}^{(k)} \right> - \left<\hat\matX{}^{(i)},\bar\w{}^{(k)} \right>}^2
\\
&& +
\sabs{\left<\hat\matX{}^{(i)},\bar\w{}^{(k)} \right> - \left<\hat\matX{}^{(i)},\tilde\w{}^{(k)} \right>}^2)
\\
&\leq&
c_1^2\cdot\FsNorm{\bar\matX{}^{(i)} - \tilde\matX{}^{(i)}}^2\FsNorm{\bar\w{}^{(k)}}^2
+c_1^2\cdot\FsNorm{\tilde\matX{}^{(i)} - \hat\matX{}^{(i)}}^2\FsNorm{\bar\w{}^{(k)}}^2
\\
&& +
c_1^2\cdot\FsNorm{\hat\matX{}^{(i)}}^2\FsNorm{\bar\w{}^{(k)} - \tilde\w{}^{(k)}}^2
\end{eqnarray*}
Considering all the elements, for some constant $c_u>0$, we derive
\begin{eqnarray*}
\FsNorm{\z - \tilde \z}^2
&\leq&
\sum_i\sum_k \sabs{z_i{}^{(k)} - \hat z_i{}^{(k)}}^2
\\
&\leq&
c_u^2\cdot\FsNorm{\bar\matX- \tilde\matX}^2\FsNorm{\bar\matW}^2
+
c_u^2\cdot\FsNorm{\tilde\matX- \hat\matX}^2\FsNorm{\bar\matW}^2
+
c_u^2\cdot\FsNorm{\hat\matX}^2\FsNorm{\bar\matW- \tilde\matW}^2
\end{eqnarray*}
From above we have,
\begin{eqnarray}\label{eqn:conv_upper_bound}
\FsNorm{\z - \tilde \z}
&\leq&
c_u \cdot (\FsNorm{\bar\matX- \tilde\matX}\FsNorm{\bar\matW}
+
\FsNorm{\tilde\matX- \hat\matX}\FsNorm{\bar\matW}
+
\FsNorm{\hat\matX}^2\FsNorm{\bar\matW- \tilde\matW})
\end{eqnarray}
Note that
$$
0 \leq \sabs{z_i^{(k)}}  = \sabs{\left<\bar\matX{}^{(i)},\bar\w{}^{(k)} \right>}
\leq c_1\cdot \FsNorm{\bar\matX{}^{(i)}}\FsNorm{\bar\w{}^{(k)}} ,
$$
and $\sabs{z_i^{(k)}}$ is close to zero when the $k$th ofm filters are orthogonal to input image patch. 
For pre-trained, locally optimal weights we expect that not all the elements of an ofm are  close to zero. For simplicity of analysis we assume that, for some constant $c^{(k)}>0$
$$
\FsNorm{\z^{(k)}}^2
=\sum_i \sabs{z_i^{(k)}}^2 
\geq
(c^{(k)})^2\FsNorm{\bar\matX}^2\FsNorm{\bar\w{}^{(k)}}^2
$$
and, for $c_{min} = \min\{c^{(k)}\}$
\begin{eqnarray}\label{eqn:conv_lower_bound}
\FsNorm{\z}^2
=
\sum_k\FsNorm{\z^{(k)}}^2
\geq
(c_{min})^2
\sum_k\FsNorm{\bar\matX}^2\FsNorm{\bar\w{}^{(k)}}^2
\geq
(c_{min})^2\FsNorm{\bar\matX}^2\FsNorm{\bar\matW}^2
\end{eqnarray}
Combining (\ref{eqn:conv_upper_bound}) and (\ref{eqn:conv_lower_bound}) we have, for some constant $c>0$,
$$
\frac{\FsNorm{\z - \tilde \z}}{\FsNorm{\z}}
\leq 
c\cdot \frac{\FsNorm{\bar\matX- \tilde\matX}}{\FsNorm{\bar\matX}}
+
c\cdot \frac{\FsNorm{\tilde\matX- \hat\matX}}{\FsNorm{\bar\matX}}
+
c\cdot\frac{\FsNorm{\hat\matX}}{\FsNorm{\bar\matX}}\frac{\FsNorm{\bar\matW- \hat\matW}}{\FsNorm{\bar\matW}}
$$
\\
$\bullet$ \textbf{Matrix Multiplication (Parametric)}:
\\
For an $m$ dimensional input $\x$ and $d\times m$ weight matrix $\matW$ along with $d$ dimensional bias $\b$, where $d$ is the number of classes, output $\y$ can be written as a linear operation, 
$$
\y = \matW\x+\b = \bar\matW\bar\x
$$
Similarly, perturbed output is 
$
\tilde\y = \tilde\matW\hat\x.
$
Note that, for some constant $0\leq c_u\leq 1$
$$
\FsNorm{\y} \leq c_u\cdot \FsNorm{\bar\matW}\FsNorm{\bar\x}
$$
Then, we can derive 
$$
\FsNorm{\y-\tilde\y}
\leq
c_u\cdot\FsNorm{\bar\x-\tilde\x}\FsNorm{\bar\matW}
+
c_u\cdot\FsNorm{\tilde\x-\hat\x}\FsNorm{\bar\matW}
+
c_u\cdot\FsNorm{\hat\x}\FsNorm{\bar\matW-\tilde\matW}
$$
For locally optimal pre-trained weights the output is expected to be not in a neighborhood of zero. Thus, we make the following assumption, for some constant $c_0>0$,
$$
c_0 \cdot\FsNorm{\bar\matW}\FsNorm{\bar\x}\leq \FsNorm{\y}
$$
Combining the above results, for some constant $c>0$, we get
$$
\frac{\FsNorm{\y-\tilde\y}}{\FsNorm{\y}}
\leq
c\cdot \frac{\FsNorm{\bar\x-\tilde\x}}{\FsNorm{\bar\x}} 
+ 
c\cdot \frac{\FsNorm{\tilde\x-\hat\x}}{\FsNorm{\bar\x}}
+
c\cdot \frac{\FsNorm{\hat\x}}{\FsNorm{\bar\x}}\frac{\FsNorm{\bar\matW-\tilde\matW}}{\FsNorm{\bar\matW}}
$$
$\bullet$ \textbf{Pooling (Non-Parametric)}:
\\
We can interpret pooling as performing a dot product between input patch and an $s\times r$ filter that contains constant entries. For example, for max pooling the filter has exactly one non-zero (value is 1) corresponding to the maximum element of the input patch, and others are zeros; and for average pooling all the entries of the filter are $\frac{1}{s\cdot r}$. We can create multiple such filters along input feature map (ifm) dimensions (for max pooling the position of non-zero entry may change). Denoting this `fake' bunch of filters as $\matW$, we have
$$
\FsNorm{\matW} = c_0
$$
where $c_0>0$ is a constant that depends only on the dimensions of input and pooling filter. The remaining operation is similar to convolution barring the accumulation of numbers across filters. 
We now bound the output perturbation for pooling. Let $\x$ denote the input patch on which pooling has been applied.
\\\\
$\bullet$ \textit{Max Pooling}: 
We bound the change in output of Max Pooling when it is applied on a distorted feature map input $\hat\x$.
Let $i^*=\argmax_i\{\x_i\}$,  $k^*=\argmax_k\{\tilde\x_i\}$, and $j^*=\argmax_j\{\hat\x_j\}$. Then, 
$$
y = \text{Max-Pooling}(\x) = \x_{i^*},
\quad
\tilde y = \text{Max-Pooling}(\hat\x) = \hat\x_{j^*}.
$$
\begin{eqnarray*}
\abs{y - \tilde y} 
&\leq& 
\sabs{\x_{i^*} - \tilde\x_{k^*}}
+
\sabs{\tilde\x_{k^*} - \hat\x_{j^*}}
\leq
\max_i\sabs{\x_i - \tilde\x_i}
+
\max_i\sabs{\tilde\x_i - \hat\x_i}
\\
&=& 
\MaxNorm{\x-\hat\x} + \MaxNorm{\tilde\x-\hat\x}
\leq 
c\cdot(\FsNorm{\x-\tilde\x} + \FsNorm{\tilde\x-\hat\x})
\end{eqnarray*}
for some constant  $0<c\leq 1$.
\\\\
$\bullet$ \textit{Mean Pooling}: 
We have similar analysis for mean pooling, where the input patch $\x$ is mapped to its mean.
$$
y = \text{Mean-Pooling}(\x) = \frac{1}{n}\sum_i \x_i.
$$
We bound the change in output of Mean Pooling when it is applied on a distorted feature map input $\hat\x$.
\begin{eqnarray*}
\abs{y-\tilde y} 
&\leq& 
\sabs{\sum_i\x_i - \sum_k\tilde\x_k}/n
+
\sabs{\sum_k\tilde\x_k - \sum_j\hat\x_j}/n
\\
&\leq&
\sum_i\sabs{\x_i - \tilde\x_i}/n 
+
\sum_i\sabs{\tilde\x_i - \hat\x_i}/n 
\\
&=& 
\OsNorm{\x - \tilde\x}/n
+ \OsNorm{\tilde\x - \hat\x}/n
\leq 
\frac{1}{\sqrt{n}}(\FsNorm{\x - \tilde\x}+\FsNorm{\tilde\x - \hat\x})
\end{eqnarray*}
Considering the entire feature map we have, for some constant $c_1>0$ that depends only on the dimensions of ifm and pooling region,
$$
\FsNorm{\matY - \tilde\matY}
\leq
c_1\cdot (\FsNorm{\matX - \tilde\matX}+\FsNorm{\tilde\matX - \hat\matX})
$$

Going by the same argument on pre-trained model that output of pooling layer cannot be close to  zero, we assume that output tensor $\matY$ satisfies,
$$
\FsNorm{\matY} 
\geq c_0\cdot \FsNorm{\matX}
$$
where $c_0>0$ is a universal constant.
Combining the above two
$$
\FsNorm{\matY - \tilde\matY}/\FsNorm{\matY}
\leq c\cdot (\FsNorm{\matX - \tilde\matX} + \FsNorm{\tilde\matX - \hat\matX})/\FsNorm{\matX}
$$

\noindent
$\bullet$  \textbf{ReLU (Non-Parametric)}:
\\
ReLU on input number $x$ is defined as
\quad 
$ReLU(x) := \max\{0,x\}$.
\\
Output of ReLU (element-wise) for input feature map $\matX$ is  
$\matY = ReLU(\matX)$. 
We bound the change in output of ReLU for a perturbed input $\hat x$.
Let 
$$
h = \abs{\text{ReLU}(x) - \text{ReLU}(\hat x)} = \abs{\max\{0,x\} - \max\{0,\hat x\}}.
$$
We consider the following cases. 
\begin{eqnarray*}
\text{Case I}: \quad h &=& \abs{x - \hat x}, \text{for } x, \hat x \geq 0,\\
\text{Case II}: \quad h &=& \abs{x} < \abs{x-\hat x}, \text{for } x \geq 0, \hat x < 0,\\
\text{Case III}: \quad h &=& \abs{\hat x} < \abs{x-\hat x}, \text{for } x < 0, \hat x \geq 0,\\
\text{Case IV}: \quad h &=& 0 < \abs{x-\hat x}, \text{for } x < 0, \hat x < 0.
\end{eqnarray*}
The last case suggests if $x$ is negative and we perturb it to some arbitrary negative value, such change has no effect on the outcome. 

Also, note that ReLU is a Lipschitz continuous function, and the Lipschitz constants for the four cases are $1, <1, <1, 0$, respectively. Assuming all the cases are equally likely the expected value of the Lipschitz constant $c$ is $c < 3/4 = 0.75$. 
Therefore, we expect ReLU to work as a noise dampener. 

In presence of input perturbation $\tilde x$ and activation perturbation $\hat x$, we can derive, for $0<c_1<1$, 
$$
h^2 
\leq 
(\text{ReLU}(x)-\text{ReLU}(\tilde x))^2 
+ 
(\text{ReLU}(\tilde x)-\text{ReLU}(\hat x))^2
\leq
c_1^2(x-\tilde x)^2 + c_1^2(\tilde x - \hat x)^2
$$
For the entire feature map, we can write 
$$
\FsNorm{\matY - \hat\matY}
=\sqrt{\sum_i(ReLU(\matX_i)-ReLU(\hat\matX_i))^2}
\leq 
c_1\cdot 
(\FsNorm{\matX-\tilde\matX}
+ \FsNorm{\tilde\matX-\hat\matX}).
$$
Note that, 
$
\FsNorm{\matY} \leq \FsNorm{\matX}.
$
Also, we expect that the output of ReLU, for a locally optimal pre-trained network, would not be in the neighborhood of zero. Thus, we assume, for some constant $c>0$,
$$
c\cdot \FsNorm{\matX}\leq \FsNorm{\matY}
$$
For example, if all the entries of $\matX$ follow a symmetric distribution w.r.t zero, $c = 1/\sqrt{2}$.
Combining the above inequalities we derive the desired expression.

\subsubsection{Proof of Theorem \ref{theorem:residual}}
\label{sec:proof_residual}
\begin{proof}
Let $\alpha^*\hat\matW{}^*\leftarrow \text{Ternary}(\matR)$ be a ternary representation of $\matR$ using Algorithm \ref{alg:ternary}. Then, we have
$$
\alpha^* 
= 
{\left<\matR,\hat\matW{}^*\right>}/{\FsNorm{\hat\matW{}^*}^2},
\quad
\text{ and }
\quad
\alpha^*\hat\matW{}^* 
= 
\left<\matR,\hat\matW{}^*/{\FsNorm{\hat\matW{}^*}}\right>\cdot
{\hat\matW{}^*} /{\FsNorm{\hat\matW{}^*}}.
$$
By construction, the following orthogonality holds: 
$
\alpha^*\hat\matW{}^*\perp 
(\matR- \alpha^*\hat\matW{}^*).
$
It follows that
\begin{eqnarray}\label{eqn:residual_improve}
\FsNorm{\matR}^2 
=
\FsNorm{\alpha^*\hat\matW{}^*}^2
+\FsNorm{\matR- \alpha^*\hat\matW{}^*}^2
>
\FsNorm{\matR- \alpha^*\hat\matW{}^*}^2
\end{eqnarray}
That is, subtracting the ternary vector $\alpha^*\hat\matW{}^*$ from $\matR$ reduces the $\ell_2$ error. In our set up, we interpret $\matR$ as the residual error produced by the earlier ternary representations, and adding the new ternary $\alpha^*\hat\matW{}^*$ to the solution set strictly reduces the $\ell_2$ error. 

In Algorithm \ref{alg:residual} only one block is getting updated every iteration via residual ternary. So, it is sufficient to show that the $\ell_2$ error for this block gets reduced. Let the index for this block be $k^*$.
\begin{eqnarray*}
\delta^{(i-1)} &=&
\FsNorm{\matW - \sum_k\sum_t\alpha^{(t)}_{(k)}\hat\matW{}^{(t)}_{(k)}}^2
\\
&=&
\FsNorm{\sum_k\matW_{(k)} - \sum_k\sum_t\alpha^{(t)}_{(k)}\hat\matW{}^{(t)}_{(k)}}^2
\\
&=&
\sum_k\FsNorm{\matW_{(k)} - \sum_t\alpha^{(t)}_{(k)}\hat\matW{}^{(t)}_{(k)}}^2
\\
&=&
\sum_{k\neq k^*}\FsNorm{\matW_{(k)} - \sum_t\alpha^{(t)}_{(k)}\hat\matW{}^{(t)}_{(k)}}^2
+
\FsNorm{\Delta_{k^*}}^2
\\
&>&
\sum_{k\neq k^*}\FsNorm{\matW_{(k)} - \sum_t\alpha^{(t)}_{(k)}\hat\matW{}^{(t)}_{(k)}}^2
+
\FsNorm{\Delta_{k^*}-\alpha^{(\tau_{k^*}+1)}_{(k^*)}\hat\matW{}^{(\tau_{k^*}+1)}_{(k^*)}
}^2, \quad \text{from } (\ref{eqn:residual_improve})
\\
&=&
\FsNorm{\sum_{k}\matW_{(k)} - \sum_{k}\sum_t\alpha^{(t)}_{(k)}\hat\matW{}^{(t)}_{(k)}
-\alpha^{(\tau_{k^*}+1)}_{(k^*)}\hat\matW{}^{(\tau_{k^*}+1)}_{(k^*)}
}^2
\\
&=&
\FsNorm{\matW - \sum_{k}\sum_t\alpha^{(t)}_{(k)}\hat\matW{}^{(t)}_{(k)}
-\alpha^{(\tau_{k^*}+1)}_{(k^*)}\hat\matW{}^{(\tau_{k^*}+1)}_{(k^*)}
}^2
\\
&=&
\delta^{(i)}
\end{eqnarray*}
Above we use the orthogonality of FGQ blocks and (\ref{eqn:residual_improve}).
\end{proof}
\subsection{Discussion on Throughput}
$\bullet$ \textit{Compute Assessment: }
Let the cost for 8-8 ops and 8-2 ops be $C_8$ and $C_2$, respectively. For 8-8 representation, let the total number of 8-8 ops be $T$. Using FGQ ternary (no residual) with block size $N$ we essentially replace a group of $N$ 8-8 ops by $N$ 8-2 ops and 1 8-8 ops. 
Then, the total compute cost for 8-8 is $C_8\cdot T$, and that for FGQ ternary is $C_8\cdot\frac{T}{N}+C_2\cdot T$. For $r$ step residual ternary, we use $r$ additional FGQ ternary blocks, incurring a total cost $(r+1)\cdot (C_8\cdot\frac{T}{N}+C_2\cdot T)$. Therefore, gain in compute for $r$-step residual ternary over 8-8 representation is
$$
\pi_c(N,r)
=
\frac{C_8\cdot T}{(r+1)\cdot (C_8\cdot\frac{T}{N}+C_2\cdot T)}
$$
Assuming $C_8=c\cdot C_2$, for some $c>1$, we have
$$
\pi_c(N,r)
=
\frac{c}{(r+1)\cdot (\frac{c}{N}+1)}
$$
$\bullet$ \textit{Memory Bandwidth Assessment: } 
We assume that we can have one input and one output buffer (and an additional buffer for ResNet type networks), and we can use them interchangeably, i.e., output buffer of earlier layer can be used as input buffer to the next layer.   
Also, we assume that the weights are streamed from memory because the model  size often exceeds on-chip memory for most of the devices (e.g., SKX-CPU, TPU). Therefore, for bandwidth bound case, gain in $r$-step residual over 8-8 representation is 
$$
\pi_m{N, r}
\approx 
\frac{4}{(r+1)\cdot (\frac{1}{N}+1)}.
$$
For $r$-step residual using block size $N$, the gain over 8-8 is $\pi_c(N,r)$ and $\pi_m(N,r)$ for  compute bound case and memory bandwidth bound case, respectively.
For $N=64$, assuming $c\approx 5$, and for $r$-step residual ternary (e.g., $r+1=2.4$, using $2.4\times$ FGQ ternary blocks) $\pi_c(N,r) \approx 2$, and $\pi_m(N,r) \approx 1.6$.

\end{document}